\documentclass[conference,10pt,a4paper]{IEEEtran}
\usepackage{amsmath,amsthm,amssymb,bbm,enumerate}
\usepackage{graphicx,float,etoolbox}
\usepackage{xcolor,svg}
\usepackage{pgfplots,balance,microtype,etoolbox,cite }

\def \arxive {1}

\title{Secure Computation to Hide Functions of Inputs}

\theoremstyle{plain} \newtheorem{thm}{Theorem}
\theoremstyle{plain} 
\theoremstyle{plain} \newtheorem{claim}{Claim}
\theoremstyle{plain} \newtheorem{defn}{Definition}
\theoremstyle{definition} \newtheorem{exmp}{Example}
\theoremstyle{plain} 
\theoremstyle{definition} \newtheorem{remark}{Remark}
\pgfplotsset{compat=1.14}
\makeatletter
\patchcmd{\@makecaption}
  {\scshape}
  {}
  {}
  {}
\makeatother
\begin{document}
\author{%
  \IEEEauthorblockN{Gowtham~R.~Kurri and Vinod~M.~Prabhakaran}
  \IEEEauthorblockA{Tata Institute of Fundamental Research\\
                    Mumbai, India\\
                    Email: \{k.raghunath, vinodmp\}@tifr.res.in}
}

\maketitle

\begin{abstract}
We consider a two-user secure computation problem in which Alice and Bob communicate interactively in order to compute some deterministic functions of the inputs. The privacy requirement is that each user should not learn any additional information about a \emph{function} of the inputs other than what can be inferred from its own input and output. For the distribution-free setting, i.e., when the protocol must be correct and private for any joint input distribution, we completely characterize the set of all securely computable functions. When privacy is required only against Bob who computes a function based on a single transmission from Alice, we show that asymptotically secure computability is equivalent to perfectly secure computability. Separately, we consider an eavesdropper who has access to all the communication and should not learn any information about some function of the inputs (possibly different from the functions to be computed by the users) and show that interaction may be necessary for secure computation.
\end{abstract}
\section{Introduction}
The goal of two-user secure computation (see \cite{CramerDBNB15} and reference therein) is for users to interactively compute a function without revealing any additional information about their inputs other than what can be inferred from their own inputs and output. In some settings, privacy might be desired only for certain attributes of the data and not the entire input. This motivates the study of the following problem.

We consider a two-user secure computation problem in which Alice and Bob communicate interactively in order to compute some deterministic functions of the inputs. Unlike in standard secure computation, our privacy requirement is that each user should not learn any additional information about a \emph{function} of the inputs other than what can be inferred from its own input and output. Separately, we consider an eavesdropper who has access to all the communication and should not learn any information about some function of the inputs (possibly different from the functions to be computed by the users).  

Two-user secure computation was formally introduced in the context of computational secrecy by Yao in his seminal works ~\cite{Yao82,Yao86}. Information theoretically, secure computation among $n$ users was studied by Ben-Or, Goldwasser, and Wigderson~\cite{BGW88} and Chaum, Crep\`{e}au, and Damg{\aa}rd~\cite{CCD88} independently. They showed that any function can be securely computed even if $t<\frac{n}{2}$ ($t<\frac{n}{3}$, resp.) honest-but-curious (malicious, resp.) users collude. 

Not all functions are securely computable by two users interacting over a noiseless link\footnote{Note that if the
two users have access to a noisy channel or correlated random variables, a
larger class of functions may be securely computed~\cite{CREK88}. A characterization of such stochastic resources which allow any function to be securely computed is given in \cite{MajiPR12}.} (e.g., binary \text{AND} function is not securely computable by two users). Beaver~\cite{Beaver89} and Kushilevitz~\cite{Kuchilevitz} gave a combinatorial characterization of two-user securely computable deterministic functions. Maji et al.~\cite{MajiPR09} showed that the same characterization holds for statistically secure protocols also. Narayan et al.~\cite{NarayanTW15} gave an alternate characterization using common randomness generated by interactive deterministic secure protocols. Data et al.~\cite{DataKRP18} gave single-letter characterizations of feasibility and optimal communication rates for two-user interactive secure randomized function computation.

 Tyagi et al.~\cite{TyagiNG11}, Tyagi~\cite{Tyagi2012}, and Gohari~et al.~\cite{GohariYA12} studied secure computation when a certain function of inputs to the users needs to be hidden from an eavesdropper having access to the communication. Lindell et al.~\cite{Lindell} studied input-size hiding secure computation. Basciftci et al.~\cite{BasciftciWI16} and Kalantari et al.~\cite{KalantariSS18} studied privacy versus utility trade-offs in data release mechanisms where the input data comprises of some useful data (whose information revealed by the output data quantifies utility) and some sensitive data (whose information revealed by the output data quantifies privacy leakage).

Our main results are as follows:
\begin{itemize}
\item In our two-user secure computation problem, for the distribution-free setting, i.,e., when the protocol must be correct and private for any joint input distribution (see Claim~\ref{claim}), we completely characterize the set of all securely computable functions (Theorem~\ref{theorem:kush}). 
\item When privacy is required only against Bob who computes a function based on a single transmission from Alice, we show that asymptotically secure computability is equivalent to perfectly secure computability (Theorem~\ref{equivalence}). 
\item When privacy is required against an eavesdropper, Tyagi et al.~\cite{TyagiNG11} studied a special case where the function that needs to be hidden from the eavesdropper is same as the one computed by the users and showed that interaction is not necessary for secure computation. Later, Tyagi \cite{Tyagi2012} studied a larger class of functions where the function that needs to be hidden from an eavesdropper is equal to one of the functions computed by the users and gave a characterization of secure computability for a class of functions. The protocol used for achievability there involves interactive communication. We ask the complementary question: Is interaction necessary for secure computation? We answer this question in the affirmative through an example (Example~\ref{example}). 
\end{itemize}

\section{Privacy Against Users Themselves}
\subsection{Distribution-Free Setting}
\begin{figure}[htbp]
\begin{center}
\includegraphics[scale=0.9]{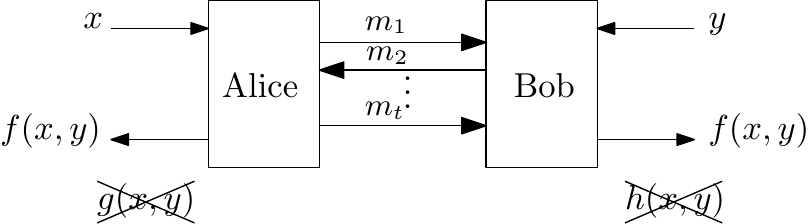}
\end{center}
\caption{Secure computation with privacy against users themselves. Alice and Bob communicate interactively in order to compute $f(x,y)$ such that Alice (resp.Bob) should not learn any additional information about $g(x,y)$ (resp. $h(x,y)$) other than what can be inferred from her (resp. his) input and $f(x,y)$.}\label{distfree}
\end{figure}
Consider two honest-but-curious (semi-honest) users Alice and Bob having inputs $x$ and $y$ respectively. They communicate interactively over many rounds to compute a function $f(x,y)$ privately, that is, Alice (resp. Bob) should not learn any additional information about $g(x,y)$ (resp. $h(x,y)$) other than what can be inferred from her (resp. his) input and the function value $f(x,y)$ (see Figure~\ref{distfree}). The users alternately send messages to each other, i.e., the message $m_i$ that a user sends in the $i^{\text{th}}$ round is a function of its input, all the messages it has seen so far $m_1,m_2,\dots,m_{i-1}$, and its private randomness. Let $m$ denote the communication string of messages $m_1,m_2,\dots,m_t$ that are sent during the protocol. We say that the protocol $\Pi$ computes the function $f(x,y)$ correctly\footnote{While the definition involves any probability of error strictly less than 0.5, the protocol presented in the proof of Theorem~\ref{theorem:kush} achieves perfect correctness, i.e., $P(\Pi(x,y)=f(x,y))=1$.} if $P(\Pi(x,y)=f(x,y))>0.5$. We say that a protocol for computing $f(x,y)$ is private against Alice with respect to $g(x,y)$ if for every two inputs $(x,y_1)$ and $(x,y_2)$ satisfying $g(x,y_1)\neq g(x,y_2)$ and $f(x,y_1)=f(x,y_2)$, and for every communication string $m$, $p(m|x,y_1)=p(m|x,y_2)$. Similarly, we say that a protocol for computing $f(x,y)$ is private against Bob with respect to $h(x,y)$ if for every two inputs $(x_1,y)$ and $(x_2,y)$ satisfying $h(x_1,y)\neq h(x_2,y)$ and $f(x_1,y)=f(x_2,y)$, and for every communication string $m$, $p(m|x_1,y)=p(m|x_2,y)$. Now, the following claim gives an alternative definition for the privacy..
\begin{claim}\label{claim}
A protocol $\Pi$ is private against Alice and Bob with respect to the functions $f(x,y)$ and $h(x,y)$ if and only if for all the input distributions $p_{XY}$, we have $I(M;G|F,X)=I(M;H|F,Y)=0$, where $F=f(X,Y), G=g(X,Y)$, and $H=h(X,Y)$.
\end{claim}
This claim is proved in\if \arxive 0 an extended version~\cite{KurriP19}\fi \if \arxive 1 the Appendix\fi.

A function triple $(f(x,y),g(x,y),h(x,y))$ is said to be \emph{securely computable} if there exists a protocol computing $f(x,y)$ that is private against Alice and Bob with respect to the functions $g(x,y)$ and $h(x,y)$, respectively. We recover the standard two-user secure computation problem~\cite{Kuchilevitz} by setting $g(x,y)=y$ and $h(x,y)=x$. 

We define $x_1\sim_{\mathcal{X}} x_2$ if $\exists y$ s.t. $g(x_1,y)\neq g(x_2,y)$ and  $f(x_1,y)=f(x_2,y)$. An equivalence relation $\equiv_\mathcal{X}$ on $\mathcal{X}$ is defined as the reflexive, transitive closure of the relation $\sim_\mathcal{X}$. Similarly, the relations $\sim_\mathcal{Y}$ and $\equiv_\mathcal{Y}$ are defined on $\mathcal{Y}$ also. Let {$R$} be the matrix corresponding to $f(x,y)$ with entry in row $x$ and column $y$ equal to $f(x,y)$. A matrix is called \emph{monochromatic} if {all its entries are equal}. 
{
\begin{defn}
A matrix is called \emph{forbidden} (for functions $g$ and $h$) if it is not monochromatic, all $x$'s are equivalent (under $\equiv_\mathcal{X}$), and all $y$'s are equivalent (under $\equiv_\mathcal{Y}$).
\end{defn}
}
 Then, we have the following theorem which characterizes the set of all securely computable function triples.

\begin{thm}\label{theorem:kush}
A function triple $(f(x,y),g(x,y),h(x,y))$ is securely computable if and only if the corresponding matrix {$R$} does not contain a forbidden sub-matrix.
\end{thm} 
Trivially, for constant functions $g(x,y)=h(x,y)=c$, every function $f(x,y)$ is securely computable with respect to the functions $g(x,y)$ and $h(x,y)$. Here we give a non-trivial example.
\begin{exmp}
In this example, we give a non-trivial function $f(x,y)$ which is not securely computable under the standard notion of privacy~\cite{Kuchilevitz} but is securely computable under the weaker notion of privacy with respect to some functions $g(x,y)$ and $h(x,y)$. The matrix corresponding to a function $f(x,y)$ is given in Table~\ref{table}. {Under the standard notion of privacy~\cite{Kuchilevitz}, i.e., when $g(x,y)=y$ and $h(x,y)=x$, $f(x,y)$ is not securely computable because the matrix itself is a forbidden matrix.}
\begin{table}[htbp]
\caption{A function triple $(f(x,y),\tilde{g}(y),\tilde{h}(x))$.}\label{table}
\begin{center}
\begin{tabular}{c|c|c c c|c}
& & &$\mathcal{Y}$ & &\\
\hline
&$f(x,y)$ & $y_1$ & $y_2$ & $y_3$&{\color{blue}{$\tilde{h}(x)$}}\\
\hline
&$x_1$ & $0$ & $0$ & $1$& {\color{blue}$1$}\\
$\mathcal{X}$&$x_2$ & $0$ & $1$ & $1$ & {\color{blue}$1$}\\
& $x_3$ & $2$ & $1$ & $0$ & {\color{blue}$2$}\\
\hline
& {\color{purple}$\tilde{g}(y)$} & {\color{purple}$1$} & {\color{purple}$2$} & {\color{purple}$2$}
\end{tabular}
\end{center}
\end{table}
{However, when $g(x,y)=\tilde{g}(y)$  and $h(x,y)=\tilde{h}(x)$ where the functions $\tilde{g}(y)$ and $\tilde{h}(x)$ are shown in Table~\ref{table}}, by Theorem~\ref{theorem:kush}, the function triple $(f(x,y),g(x,y),h(x,y))$ is securely computable. 
\end{exmp}
\begin{proof}[Proof of Theorem~\ref{theorem:kush}]
The proof leverages and extends the proof techniques of ~\cite[Theorems 1 and 2]{Kuchilevitz}. {Recall that 
\begin{multline*}
p(m_1,m_2,\dots,m_t|x,y)=\\ p(m_1|x)p(m_2|y,m_1)\dots p(m_{t-1}|x,m_1,\dots,m_{t-2})\\ \times p(m_t|y,m_1,\dots,m_{t-1}),
\end{multline*}
where it is assumed that $t$ is even, without loss of generality.}

\underline{`Only if' part}:
Suppose the matrix {$R$} corresponding to $(f(x,y),g(x,y),h(x,y))$ has a forbidden sub-matrix $N=\{x_1,\dots,x_l \}\times \{y_1,\dots,y_k \}$ and that there exists a secure protocol. Without loss of generality, the last message sent in the protocol $\Pi$ is assumed to contain the function computed and is denoted by $\Pi(x,y)$. Since $N$ is forbidden, all $x_i$'s are equivalent and can be ordered in a way so that for every $i>1$ there exists $i^\prime<i$ such that $x_i\sim_\mathcal{X} x_{i^\prime}$. A similar ordering exists on $y_j$'s also. We shall show that $p(m_1,\dots,m_t|x_i,y_j)=p(m_1,\dots,m_t|x_1,y_1)$ for every $(x_i,y_j)\in N$. Then, the probability distribution of communication messages is same for any $(x_i,y_j)\in N$, which, in turn implies that the last message (i.e., $f(x,y)$) of the communication message is distributed in the same way. Since $N$ is not monochromatic this contradicts the correctness of the protocol, $P(\Pi(x_i,y_j)=f(x_i,y_j))>0.5$, for $(x_i,y_j)\in N$. It now remains to show that $p(m_1,\dots,m_t|x_i,y_j)=p(m_1,\dots,m_t|x_1,y_1)$ for every $(x_i,y_j)\in N$. We show this by induction on the number of rounds $t$. For $i>1$, there is some $i^\prime<i$ such that $x_i\sim_\mathcal{X} x_{i^\prime}$. So, there exists $y$ s.t. $g(x_i,y)\neq g(x_{i^\prime},y)$ and $f(x_i,y)=f(x_{i^\prime},y)$. Now, from the privacy condition $p(m_1|x_i)=p(m_1|x_i,y)=p(m_1|x_{i^\prime},y)=p(m_1|x_{i^\prime})$. Repeating the same argument as above if $i^\prime>1$, we get $p(m_1|x_i)=p(m_1|x_{i^\prime})=\dots =p(m_1|x_1)$. So, the base case for $t=1$ is true. Assume that $$p(m_1,\dots,m_{t-1}|x_i,y_j)=p(m_1,\dots,m_{t-1}|x_1,y_1)$$ for every $(x_i,y_j)\in N$. If $$p(m_1,\dots,m_{t-1}|x_i,y_j)=p(m_1,\dots,m_{t-1}|x_1,y_1)=0$$ then it is trivial that $$p(m_1,\dots,m_{t}|x_i,y_j)=p(m_1,\dots,m_{t}|x_1,y_1)=0.$$ So, assume that $$p(m_1,\dots,m_{t-1}|x_i,y_j)=p(m_1,\dots,m_{t-1}|x_1,y_1)\neq 0.$$ Without loss of generality, assume that $t$ is even (odd $t$ can be handled similarly). Note that
 \begin{multline}
 p(m_1,\dots,m_t|x_i,y_j)=p(m_1,\dots,m_{t-1}|x_i,y_j)\\
 \times p(m_t|y_j,m_1,\dots,m_{t-1})\label{thm1_induc}.
 \end{multline} 
Since there is an ordering on $y_j$'s, for $j>1$, there is some $j^\prime<j$ such that $y_j\sim_\mathcal{Y} y_{j^\prime}$. So, there exists $x$ s.t. $h(x,y_j)\neq h(x,y_{j^\prime})$ and $f(x,y_j)=f(x,y_{j^\prime})$. Therefore, by \eqref{thm1_induc} and the privacy condition, $$p(m_t|y_j,m_1,\dots,m_{t-1})=p(m_t|y_{j^\prime},m_1,\dots,m_{t-1}).$$ Repeating the same argument as above if $j^\prime>1$, we get $$p(m_t|y_j,m_1,\dots,m_{t-1})=p(m_t|y_1,m_1,\dots,m_{t-1}).$$ Hence, we have $p(m_1,\dots,m_t|x_i,y_j)=p(m_1,\dots,m_t|x_1,y_1)$ for every $(x_i,y_j)\in N$. This completes the `only if' part of the proof.

\underline{`If' part}: We need the following definitions. A $C\times D$-matrix, $C\subseteq \mathcal{X}, D\subseteq \mathcal{Y}$, is \emph{rows decomposable} if there exist nonempty sets $C_1,\dots,C_t\  (t\geq 2)$ such that
\begin{itemize}
\item $C_1,\dots,C_t$ form a partition of $C$.
\item For every $x_1,x_2\in C$, if $x_1\sim_\mathcal{X} x_2$ then $x_1$ and $x_2$ are in the same $C_i$.
\end{itemize}
 {Similarly, we can define \emph{columns decomposability}}. A $C\times D$-matrix $K$ is \emph{decomposable} if one of the following conditions holds:
\begin{itemize}
\item  {$K$} is monochromatic.
\item $K$ is rows decomposable to submatrices {$C_1\times D$-matrix $A_1$,$\dots,C_t\times D$-matrix $A_t$}, which are all in turn decomposable.
\item $K$ is columns decomposable to submatrices {$C\times D_1$-matrix $B_1$,$\dots,C\times D_t$-matrix $B_t$}, which are all in turn decomposable.
\end{itemize}
Note that for a $C\times D$-matrix, the optimal row decomposition and the optimal column decomposition (optimal in the sense of maximum number of subsets in the partition) are unique and are determined by the equivalence classes under $\equiv_\mathcal{X}$ and $\equiv_\mathcal{Y}$, respectively. 
Suppose the $\mathcal{X}\times\mathcal{Y}$-matrix {$R$} corresponding to the function $f$ does not contain a forbidden sub-matrix. This implies that {$R$} is a decomposable matrix. Consider the following protocol. {Let us assume that $R$ is columns decomposable so that Alice starts the protocol.}
\begin{enumerate}[(1)]
\item Initialize $C=\mathcal{X}, D=\mathcal{Y}$, and the matrix $K=R$,
\item While the $C\times D$-matrix $K$ is non-monochromatic,
\begin{enumerate}[(a)]
\item Alice sends $i$ such that $x\in C_i$, {where the submatrices $C_1\times D$-matrix $A_1$, $\dots$, $C_t\times D$-matrix $A_t$, form the optimal rows decomposition of the $C\times D$-matrix $K$. Both users then set $C=C_i$ and $K=A_i$.}
\item If the $C\times D$-matrix $K$ is non-monochromatic, Bob sends $j$ such that $y\in D_j$, {where the submatrices $C\times D_1$-matrix $B_1$, $\dots$, $C\times D_t$-matrix $B_t$, form the optimal columns decomposition of the $C\times D$-matrix $K$. Both users then set $D=D_j$ and $K=B_j$.}
\end{enumerate}
\item Alice or Bob sends the constant value in the $C\times D$-matrix {$K$} as the value of $f(x,y)$.
\end{enumerate}
Now, we analyse the protocol for correctness and privacy. Since the $C\times D$-matrix is always decomposable and the input always belongs to the $C\times D$-matrix in every round, the algorithm terminates with the correct value of $f(x,y)$. As the protocol above is deterministic, to show that the protocol is private against Alice, {it suffices to argue that the same communication messages are transmitted for two inputs $(x,y_1)$ and $(x,y_2)$ satisfying $g(x,y_1)\neq g(x,y_2)$ and $f(x,y_1)=f(x,y_2)$. This is in fact true since the protocol never differentiates between $y_1$ and $y_2$ because $y_1\sim_\mathcal{Y}y_2$ in every iteration}. Similarly, we can argue that the protocol is also private against Bob. This completes the `if' part of the proof.  
\end{proof}
\begin{remark}
In standard secure computation~\cite{Kuchilevitz}, secure computability with any full-support input distribution (e.g., uniform distribution) is equivalent to secure computability in distribution-free setting (and therefore with any other input distribution). In contrast, for our secure computation problem, it turns out that secure computability with a full-support input distribution, e.g., the uniform distribution, does not imply secure computability with all other distributions (let alone in the distribution-free setting).\if \arxive 0 See the extended version~\cite{KurriP19} for an example which illustrates this.\fi \if \arxive 1 The following example illustrates this.
\begin{exmp}
Consider the following function $f(x,y)$ with $x\in\{0,1\}$ and $y=(y^\prime,y^{\prime \prime})\in\{0,1\}^2$ in Table~\ref{table2}.
\begin{table}[htbp]
\caption{$f\left(x,\left(y^\prime,y^{\prime\prime}\right)\right)=x\wedge y^\prime$.}\label{table2}
\begin{center}
\begin{tabular}{|c|c|c|c|c|}
\hline
$f(x,y)$ &$y=(0,0)$ &$y=(0,1)$ &$y=(1,0)$ &$y=(1,1)$\\
\hline
$x=0$ &$0$ &$0$ &$0$ &$0$\\
\hline 
$x=1$ &$0$ &$0$ &$1$ &$1$\\
\hline
\end{tabular}
\end{center}
\end{table}
Suppose $g(x,y)=y^{\prime\prime}$ needs to be hidden from Alice and $h(x,y)=x$ needs to be hidden from Bob. Consider a protocol where Bob sends $y^\prime$ to Alice who computes the output and sends it to Bob.  It is easy to check that, for uniform distribution, this protocol is secure. Now consider an input distribution $p_{XY}=p_X.p_Y$ with $X\sim \text{Uniform}\{0,1\}$ and $Y=(Y^\prime,Y^{\prime\prime})\sim \text{Uniform}\{(00),(11)\}$. With this input distribution, there does not exist a secure protocol since the problem reduces to two-user secure computation of binary AND function (as $Y^\prime=Y^{\prime\prime}$) which is impossible~\cite{Kuchilevitz}. 
\end{exmp}

\fi  

Now we turn to secure computability with a fixed input distribution. 
\end{remark}
\subsection{Non-Interactive Setting With Input Distribution}
\begin{figure}[htbp]
\begin{center}
\includegraphics[scale=0.9]{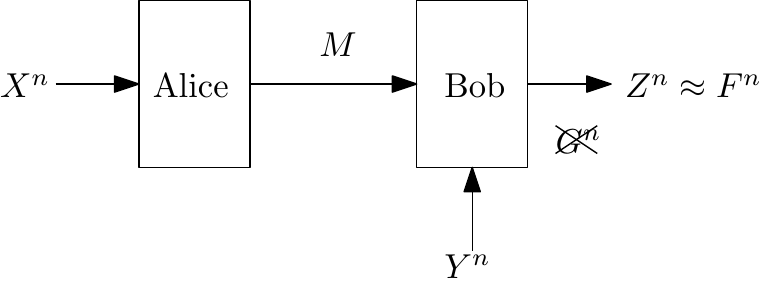}
\end{center}
\caption{Inputs are i.i.d. $p_{XY}$. Asymptotically secure setting is shown. Bob has to compute $Z^n$ that is \emph{approximately} close to $F^n$, $F_i=f(X_i,Y_i)$, while learning no additional information about $G^n$, $G_i=g(X_i,Y_i)$.}\label{figure:nonint}
\end{figure}
Let the inputs to Alice and Bob are $X^n$ and $Y^n$, respectively, where $(X_i,Y_i)$, $i=1,\dots , n$, are independent and identically distributed (i.i.d.) with distribution $p_{XY}$, {where $X$ and $Y$ take values in finite sets $\mathcal{X}$ and $\mathcal{Y}$, respectively}. In this section, we restrict ourselves to the protocols that use single round of communication. Assume that privacy is required only against Bob who computes a function based on a single transmission from Alice. We say that a triple $(p_{XY},f(x,y),g(x,y))$ is \emph{asymptotically securely computable with no interaction} if there exists a sequence of conditional probability distributions $p_{M,Z^n|X^nY^n}:=p_{M|X^n}p_{Z^n|M,Y^n}$ (see Figure~\ref{figure:nonint}) s.t. for every $\epsilon>0$,
 \begin{align}
 &P(Z^n\neq F^n)\leq \epsilon, F_i=f(X_i,Y_i),\label{eqn:distasympcorrectness}\\
 &I(M;G^n|Z^n,Y^n)\leq n\epsilon, G_i=g(X_i,Y_i)\label{eqn:prvbob},
 \end{align}
 for large enough $n$. We say that $(p_{XY},f(x,y),g(x,y))$ is perfectly securely computable if there exists $p_{U|XY}$ s.t.
 \begin{align}
 &U-X-Y,\label{eqn:chap421}\\
 &H(F|U,Y)=0,\label{eqn:chap422}\\
 &I(U;G|F,Y)=0\label{eqn:chap423}.
 \end{align}
  Instead of functions of inputs, suppose if the inputs and outputs of users need to be hidden, our earlier work~\cite[Theorem 1]{DataKRP18} shows that asymptotically secure computability is equivalent to perfectly secure computability even with interactive communication. However,  this is not true for all secure computation problems. In particular, as \cite[Remark 2]{DataKRP18} shows, this is not true for the problem of function computation with privacy against an eavesdropper studied by Tyagi et al.~\cite{TyagiNG11}. The following theorem shows that such an equivalence holds for non-interactive communication setting considered here.    
 \begin{thm}\label{equivalence}
 When privacy is required only against Bob who computes a function based on a single transmission from Alice, asymptotically secure computability is equivalent to perfectly secure computability. 
 \end{thm}
 \begin{remark}
 The problem of whether asymptotically secure computability is equivalent to perfectly secure computability for multiple rounds of communication remains open. 
 \end{remark}
 \begin{proof}
 It is easy to see that perfectly secure computability implies asymptotically secure computability since \eqref{eqn:chap421}-\eqref{eqn:chap423} define a protocol with $n=1$ and $\epsilon=0$. For the other direction, we single-letterize the constraints of asymptotically secure computability to get \eqref{eqn:chap421}-\eqref{eqn:chap423}. {Notice that the joint distribution of all the random variables is given by 
 \begin{multline}
 p_{X^n,Y^n,G^n,F^n,M,Z^n}(x^n,y^n,g^n,f^n,m,z^n)\\=\left(\prod_{i=1}^np_{XY}(x_i,y_i)\right)p_{M|X^n}(m|x^n)p_{Z^n|M,Y^n}(z^n|m,y^n)\\
 \times \mathbbm{1}\{g_i=g(x_i,y_i),f_i=f(x_i,y_i),i\in[1:n]\}.
 \end{multline}
 }
 Let $T$ be a random variable uniformly distributed on $[1:n]$ and independent of everything else. Privacy condition against Bob, \eqref{eqn:prvbob}, implies that
 \begin{align}
 n\epsilon&\geq H(G^n|Z^n,Y^n)-H(G^n|M,Z^n,Y^n)\nonumber\\
 &\geq \sum_{i=1}^n [H(G_i|F_i,Y_i)-\epsilon_1]-\sum_{i=1}^nH(G_i|M,Z^n,Y^n,G^{i-1})\label{eqn:distfree1}\\ 
 & \geq \sum_{i=1}^n[H(G_i|Z_i,Y_i)-\epsilon_1-\epsilon_2]\nonumber\\
 &\hspace{1cm}-\sum_{i=1}^n H(G_i|M,Y^{i-1},Y^{i+1:n},Z_i,Y_i)\label{eqn:distfree2}\\
 &= \sum_{i=1}^nI(U_i;G_i|Z_i,Y_i)-n(\epsilon_1+\epsilon_2)\label{eqn:distfreeaux}\\
 &=n[I(U_T,T;G_T|Z_T,Y_T)-I(T;G_T|Z_T,Y_T)]\nonumber\\
 &\hspace{1cm}-n(\epsilon_1+\epsilon_2)\nonumber\\
 &\geq nI(U_T,T;G_T|Z_T,Y_T)-n(\epsilon_1+\epsilon_2+\epsilon_3)\label{eqn:equivfinal}.
 \end{align}
 We have used the following fact in \eqref{eqn:distfree1}, \eqref{eqn:distfree2}, and \eqref{eqn:equivfinal}: if two random variables $A$ and $A'$ with 
  same support set $\mathcal{A}$  satisfy $||p_{A} -  p_{A'}||_{1} \leq \epsilon \leq 1/4 $, then it follows from \cite[Theorem 17.3.3]{CoverJ06} that $|H(A) - H(A')|\leq \eta \log |\mathcal{A}|$, where $\eta \rightarrow 0$ as $\epsilon \rightarrow 0$. Now \eqref{eqn:distasympcorrectness} implies that $\lVert p_{X^n,Y^n,Z^n,G^n}-p_{X^n,Y^n,F^n,G^n} \rVert_1\leq \epsilon$ which in turn implies  $\lVert p_{X_TY_TZ_TG_T}-p_{XYFG} \rVert_1\leq \epsilon$ using \cite[Lemma VI.2]{Cuff13}. These imply \eqref{eqn:distfree1}, \eqref{eqn:distfree2}, and \eqref{eqn:equivfinal} with $\epsilon_1,\epsilon_2,\epsilon_3\rightarrow 0$ as $\epsilon\rightarrow 0$. \eqref{eqn:distfreeaux} follows by defining $U_i=(M,Y^{i-1},Y^{i+1:n})$.
  
   From the asymptotically secure protocol, we have the Markov chain $Z^n-(M,Y^n)-X^n$ which implies that $Z^n-(M,Y^n)-(X^n,F^n)$ because $F^n$ is a deterministic function of $(X^n,Y^n)$. This further implies that $Z_i-(U_i,Y_i)-F_i, i\in[1:n]$.
 Now using Fano's inequality \cite[Theorem 2.10.1]{CoverJ06}, this implies that, for $i\in[1:n]$, 
  \begin{align}
  H(F_i|U_i,Y_i)&\leq P(Z_i\neq F_i)\leq P(Z^n\neq F^n)\leq \epsilon,
  \end{align}
  where the last inequality follows from \eqref{eqn:distasympcorrectness}. This gives us 
  \begin{align}H(F_T|U_T,T,Y_T)\leq \epsilon.\label{eqn:equivforgot} 
  \end{align} 
  Consider
  \begin{align}
 & I(M,Y^{i-1},Y^{i+1:n};Y_i|X_i)\nonumber\\
  &\leq I(M,Y^{i-1},Y^{i+1:n},X^{i-1},X^{i+1:n};Y_i|X_i)\nonumber\\
  &=I(M;Y_i|Y^{i-1},Y^{i+1:n},X^n)\label{eqn:equivasymptindep}\\
  &\leq I(M;Y^n|X^n)=0,
  \end{align}
  where \eqref{eqn:equivasymptindep} follows because $(X_i,Y_i)$ is independent of $(X^{i-1},X^{i+1:n},Y^{i-1},Y^{i+1:n})$, and the last equality follows from the Markov chain $M-X^n-Y^n$. This gives us 
  \begin{align}
(U_T,T)-X_T-Y_T\label{eqn:equivwynerziv}
  \end{align}
  Now using the continuity of total variation distance and mutual information in the probability simplex along similar lines as \cite[Lemma 6]{YassaeeGA15}, \eqref{eqn:chap421}-\eqref{eqn:chap423} follow from \eqref{eqn:equivfinal}, \eqref{eqn:equivforgot}, and \eqref{eqn:equivwynerziv}, respectively.
 \end{proof}

\section{Privacy Against an Eavesdropper}
\begin{figure}[htbp]
\begin{center}
\includegraphics[scale=0.857]{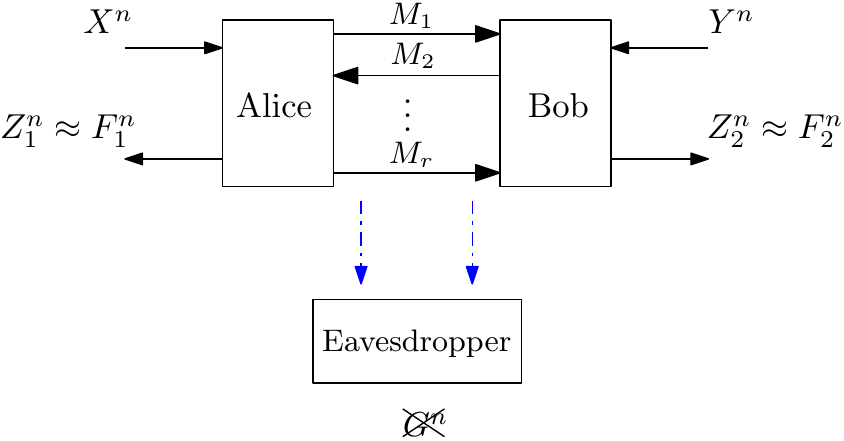}
\end{center}
\caption{Inputs are i.i.d. $p_{XY}$. Asymptotically secure setting is shown. Alice and Bob communicate interactively in order to compute $Z_1^n$ and $Z_2^n$ that are \emph{approximately} close to $F_1^n$ and $F_2^n$, respectively, $F_{ki}=f_k(X_i,Y_i), k\in[1:2], i\in[1:n]$, while hiding $G^n$, $G_i=g(X_i,Y_i)$ from an eavesdropper.}\label{figure:eaves}
\end{figure}
Let the inputs to Alice and Bob $X^n$ and $Y^n$, respectively, where $(X_i,Y_i)$, $i=1,\dots , n$, are i.i.d. with distribution $p_{XY}$, {where $X$ and $Y$ take values in finite sets $\mathcal{X}$ and $\mathcal{Y}$, respectively}. Let Alice and Bob communicate interactively and compute $Z_1^n$ and $Z_2^n$, respectively. Let $M_{[1:r]}$ denote the message transcript due to an interactive two-user protocol $\Pi$ and is accessible to an eavesdropper. We say that a tuple $(p_{XY},f_1(x,y),f_2(x,y),g(x,y))$ is \emph{asymptotically securely computable in $r$ rounds with privacy against an eavesdropper} (see Figure~\ref{figure:eaves}) if there exists a sequence of interactive protocols $\Pi_n$ such that, for every $\epsilon>0,$ there exists a large enough $n$ s.t.
\begin{align}
{P\left(\left(Z_1^n,Z_2^n\right)\neq \left(F_1^n,F_2^n\right)\right)}&\leq \epsilon,\label{eqn:eaves1}\\
I(M_{[1:r]};G^n)&\leq n\epsilon,\label{eqn:eaves2}.
\end{align}
We say that a tuple $(p_{XY},f_1(x,y),f_2(x,y),g(x,y))$ is perfectly securely computable in $r$ rounds if there exists a protocol $\Pi_n$ with $n=1$ such that \eqref{eqn:eaves1}-\eqref{eqn:eaves2} are satisfied with $\epsilon=0$. It is worthwhile to remark here that, in this setting, asymptotically secure computability is not equivalent to perfectly secure computability as pointed out in our earlier work \cite[Remark 2]{DataKRP18}. 

Tyagi et al.~\cite{TyagiNG11} studied asymptotically secure computability of this problem in multi-user setting for a special case where the function that needs to be hidden from the eavesdropper is same as the one computed by the users, i.e., $f_1=f_2=g$ in the current context. {They showed that interaction is not necessary for asymptotically secure computation, i.e, the protocol involves each user publicly communicating only one message each which is a randomized function of its own input.} Later, Tyagi \cite{Tyagi2012} studied a larger class of functions where the function that needs to be hidden from an eavesdropper is equal to one of the functions computed by the users and gave a characterization of secure computability for a class of functions. The protocol used for achievability there involves interactive communication. We ask the complementary question: Is interaction necessary for secure computation, in particular, are there any function tuples that are securely computable only using interactive protocols? We answer this question in the affirmative through the following example.

\begin{exmp}\label{example}
 Let $Y=(Y_0,Y_1)$ be a vector of two independent and uniformly distributed binary random variables and $X$ is another uniform binary random variable independent of $Y$. Let $p_{XY}$ denote this joint distribution. Consider $f_1(x,y)=f_2(x,y)=y_x$ and $g(x,y)=y_{\bar{x}}$, where $\bar{x}$ denotes the complement of $x$, i.e., $\bar{x}=1-x$. We show that this choice of $(p_{XY},f_1(x,y),f_2(x,y),g(x,y))$ is not asymptotically securely computable using non-interactive communication. We prove this via contradiction. Suppose there exists a non-interactive protocol that asymptotically securely computes $(p_{XY},f_1(x,y),f_2(x,y),g(x,y))$, {i.e., there exists a protocol in which each user transmits one message each that is a randomized function of its own input.} This means that, for every $\epsilon>0$, there exists $M_1,M_2$, and a large enough $n$ such that
 \begin{align}
 M_1-X^n-Y^n-M_2\label{eqn:alicemarkov},\\
H(Y_X^n|M_1,M_2,X^n)\leq n\epsilon \label{eqn:alicecorrectness},\\
H(Y_X^n|M_1,M_2,Y^n)\leq n\epsilon \label{eqn:bobcorrectness},\\
I(M_1,M_2;Y_{\bar{X}}^n)\leq n\epsilon \label{eqn:exmpsecurity},
 \end{align}
 where \eqref{eqn:alicecorrectness} and \eqref{eqn:bobcorrectness} uses Fano's inequality \cite[Theorem 2.10.1]{CoverJ06}. Notice that since we have the Markov chains $Y_X^n-(M_2,X^n)-M_1$ and $Y_X^n-(M_1,Y^n)-M_2$, \eqref{eqn:alicecorrectness} and \eqref{eqn:bobcorrectness} can be written as
 \begin{align}
 H(Y_X^n|M_2,X^n)\leq n\epsilon \label{eqn:alicecorrectness_new},\\
H(Y_X^n|M_1,Y^n)\leq n\epsilon \label{eqn:bobcorrectness_new}.
\end{align}
Let $\tilde{Y}=(\tilde{Y}_0,\tilde{Y}_1)=(\bar{Y}_0,Y_1)$. Now, since $p_{XY}=p_{\bar{X}Y}=p_{X\tilde{Y}}$ we have the following from \eqref{eqn:alicecorrectness_new} and \eqref{eqn:bobcorrectness_new}:
\begin{align}
 H(Y_{\bar{X}}^n|M_2,\bar{X}^n)\leq n\epsilon\label{eqn:alicefinal},\\
 H(\tilde{Y}_X^n|M_1,\tilde{Y}^n)\leq n\epsilon\label{eqn:bobfinal}.
\end{align}
Consider, 
\begin{align}
H(Y^n|M_2)&=H(Y^n|M_2,X^n) \label{eqn:alice1}\\
&=H(Y_0^n,Y_1^n|M_2,X^n)\nonumber\\
&=H(Y_X^n,Y_{\bar{X}}^n|M_2,X^n)\label{eqn:alice3}\\
&\leq H(Y_X^n|M_2,X^n)+H(Y_{\bar{X}}^n|M_2,X^n)\nonumber\\
&\leq 2n\epsilon \label{eqn:alice2},
\end{align}
where \eqref{eqn:alice1} follows because $(M_1,X^n)$ is independent of $(M_2,Y^n)$, \eqref{eqn:alice3} follows because conditioned on $X^n$, there is a bijection between $(Y_0^n,Y_1^n)$ and $(Y_X^n,Y_{\bar{X}})$, \eqref{eqn:alice2} follows from \eqref{eqn:alicecorrectness_new} and \eqref{eqn:alicefinal}. Next consider,
\begin{align}
H(X^n|M_1)&=H(X^n|M_1,Y^n)\label{eqn:bob1}\\
&=H(Y_X^n,\tilde{Y}_X^n|M_1,Y^n)\label{eqn:bob2}\\
&\leq H(Y_X^n|M_1,Y^n)+H(\tilde{Y}_X^n|M_1,Y^n)\nonumber\\
&\leq 2n\epsilon \label{eqn:bob3},
\end{align}
where \eqref{eqn:bob1} follows because $(M_1,X^n)$ is independent of $(M_2,Y^n)$, \eqref{eqn:bob2} follows because conditioned on $Y^n$, there is a bijection between $X^n$ and $(Y_X^n, \tilde{Y}_X^n)$, and \eqref{eqn:bob3} follows from \eqref{eqn:bobcorrectness_new} and \eqref{eqn:bobfinal}. Now $H(X^n|M_1)\leq 2n\epsilon$ and $H(Y^n|M_2)\leq 2n\epsilon$ implies that $H(Y_{\bar{X}}^n|M_1,M_2)\leq 4n\epsilon$. This is a contradiction to \eqref{eqn:exmpsecurity}. 

Notice that if we allow interactive communication, the tuple $(p_{XY},f_1(x,y),f_2(x,y),g(x,y))$ is asymptotically securely computable. To see this, note that $M_1=X^n$ and $M_2=Y_X^n$ satisfies the secrecy constraint \eqref{eqn:eaves2} (in fact with $\epsilon=0$). This completes the example.
 \end{exmp}
\if \arxive 1 In our notation, Tyagi et al.~\cite[Theorem~2]{TyagiNG11} state that $(p_{XY},g(x,y),g(x,y),g(x,y))$ is asymptotically securely computable if 
\begin{align}
H(g(X,Y))<I(X;Y)\label{eqn:tyagi}
\end{align}
 (and only if $H(g(X,Y))\leq I(X;Y)$).
We make the following observation.    
 \begin{thm}\label{lemma}
 If $p_{XY}$ and the function $g(x,y)$ are such that $H(g(X,Y))<I(X;Y)$, then the tuple $(p_{XY},f_1(x,y),f_2(x,y),g(x,y))$ is asymptotically securely computable for any functions $f_1(x,y)$ and $f_2(x,y)$.
 \end{thm}
 Tyagi et al.~\cite[Theorem~2]{TyagiNG11} showed that if \eqref{eqn:tyagi} is satisfied, then omniscience can be attained using non-interactive communication that is almost independent of $G^n$ (in the sense of \eqref{eqn:eaves2}). The proof of Theorem~\ref{lemma} follows from this and the fact that the users can then (approximately) compute any functions $F_1^n$ and $F_2^n$, respectively (even if $f_1$ and $f_2$ can be possibly different from $g$).
 \fi
\section{Acknowledgements}
We acknowledge support of the Department of Atomic Energy, Government of India, under project no. 12-R{\&}D-TFR-5.01-0500. Gowtham Kurri would like to thank Varun Narayanan for many helpful discussions on this paper. 
\if \arxive 1
\begin{appendix}\label{appendix:claim}
\begin{proof}[Proof of Claim~\ref{claim}]
Suppose a protocol is private against Alice and Bob. Fix any input distribution $p_{XY}$. Consider
\begin{align}
&p(m|f,x,g)\nonumber\\
&=\frac{p(x,m,f,g)}{p(f,x,g)}\nonumber\\
&=\frac{\sum\limits_yp(x,y,m,f,g)}{\sum\limits_{y,m}p(x,y,m,f,g)}\nonumber\\
&=\frac{\sum\limits_yp(x,y)p(m|x,y)p(f|x,y)p(g|x,y)}{\sum\limits_{m,y}p(x,y)p(m|x,y)p(f|x,y)p(g|x,y)}\nonumber\\
&=\frac{\sum\limits_{\substack{y:g(x,y)=g\\ f(x,y)=f}}p(x,y)p(m|x,y)}{\sum\limits_{\substack{m,y:g(x,y)=g\\ f(x,y)=f}}p(x,y)p(m|x,y)}\label{eqn:appendix1}
\end{align}
Without loss of generality, assume that there exists $y^\prime$ such that $f(x,y^\prime)=f$ and $g(x,y^\prime)=g^\prime\neq g$. Otherwise we trivially have $p(m|f,x,g)=p(m|f,x)$. Now, from the definition of privacy, we have $p(m|x,y)=p(m|x,y^\prime)$, for $y,y^\prime$ such that $f(x,y)=f(x,y^\prime)=f$ and $g(x,y)=g\neq g^\prime=g(x,y^\prime)$. This implies that $p(m|x,y)$ is same for every $y$ such that $f(x,y)=f$. This observation reduces the expression in \eqref{eqn:appendix1} to $p(m|x,y)$ such that $f(x,y)=f$. Thus, we have
\begin{align}
p(m|f,x,g)=p(m|x,y) \ \text{for}\ y \ \text{s.t.}\ f(x,y)=f\label{eqn:appendix2}.
\end{align}
Consider 
\begin{align}
p(m|f,x)&=\frac{\sum\limits_yp(x,y,m,f)}{\sum\limits_{m,y}p(x,y,m,f)}\nonumber\\
&=\frac{\sum\limits_{y:f(x,y)=f}p(x,y)p(m|x,y)}{\sum\limits_{m,y}p(x,y)p(m|x,y)}\nonumber\\
&=p(m|x,y) \ \text{for}\ y \ \text{s.t.}\ f(x,y)=f\label{eqn:appendix3},
\end{align}
where the last equality follows from the same observation mentioned above \eqref{eqn:appendix2}, i.e., $p(m|x,y)$ is same for every $y$ such that $f(x,y)=f$. From \eqref{eqn:appendix2} and \eqref{eqn:appendix3}, we have $p(m|f,x,g)=p(m|f,x)$, for $p(f,x,g)>0$ which is equivalent to $I(M;G|F,X)=0$. Similarly, privacy against Bob implies that $I(M;H|F,Y)=0$. 
\if \arxive 1
\balance
\fi
For the other direction, suppose that for all input distributions $q_{XY}$, there exists a unique $p_{M|XY}$ such that $I(M;G|F,X)=I(M;H|F,Y)=0$. Consider $(x,y_1)$ and $(x,y_2)$ such that $g(x,y_1)=g_1\neq g(x,y_2)=g_2$ and $f(x,y_1)=f(x,y_2)=f$. Fix a $p_{XY}$ that is supported only on $(x,y_1)$ and $(x,y_2)$. Consider 
\begin{align}
p(m|x,y_1)&=p(m|x,y_1,g_1,f)\nonumber\\
&=p(m|x,g_1,f)\label{eqn:claim1}\\
&=p(m|x,f)\label{eqn:claim2}\\
&=p(m|x,g_2,f)\label{eqn:claim3}\\
&=p(m|x,y_2,g_2,f)\label{eqn:claim4}\\
&=p(m|x,y_2)\nonumber,
\end{align}
where \eqref{eqn:claim1} and \eqref{eqn:claim4} follow because under this input distribution, $G$ is a function of $Y$ and vice versa, \eqref{eqn:claim2} and \eqref{eqn:claim3} follow because $I(M;G|F,X)=0$. Similarly, we can show that $p(m|x_1,y)=p(m|x_2,y)$ for $(x_1,y)$ and $(x_2,y)$ such that $g(x_1,y)\neq g(x_2,y)$ and $f(x_1,y)=f(x_2,y)$.  
\end{proof}
\end{appendix}

\fi
\bibliographystyle{IEEEtran}
\if \arxive 0
\balance
\fi
\bibliography{Bibliography}

\end{document}